\documentclass[11pt]{article}

\usepackage[letterpaper,margin=0.97in]{geometry}

\usepackage{bm, amsmath, amsthm, amssymb, amsfonts}

\usepackage{cite}
\usepackage{framed}
\usepackage[framemet hod=tikz]{mdframed}
\usepackage{appendix}
\usepackage{graphicx}
\usepackage{color}
\usepackage{varwidth}
\usepackage{wrapfig}
\usepackage{enumitem}

\setitemize{noitemsep,topsep=3pt,parsep=3pt,partopsep=3pt}

\usepackage{xspace}
\usepackage{xifthen}

\makeatletter
\renewcommand{\paragraph}{%
  \@startsection{paragraph}{4}%
  {\z@}{2.5ex \@plus 1ex \@minus .2ex}{-1em}%
  {\normalfont\normalsize\bfseries}%
}
\makeatother

\definecolor{darkgreen}{rgb}{0,0.5,0}
\definecolor{darkblue}{rgb}{0,0,0.8}
\usepackage{hyperref}
\hypersetup{
    unicode=false,          % non-Latin characters in Acrobat’s bookmarks
    colorlinks=true,        % false: boxed links; true: colored links
    linkcolor=darkblue,          % color of internal links (change box color with linkbordercolor)
    citecolor=darkgreen,        % color of links to bibliography
    filecolor=magenta,      % color of file links
    urlcolor=brown           % color of external links
}
\RequirePackage[]{silence}
\usepackage{algorithm}
\usepackage[noend]{algpseudocode}

\usepackage[nameinlink,capitalize]{cleveref}

\newtheorem{theorem}{Theorem}[section]

\newtheorem{lemma}[theorem]{Lemma}
\newtheorem{corollary}[theorem]{Corollary}
\newtheorem{definition}{Definition}[section]

\usepackage{mathtools}

\newcommand{\calA}{\ensuremath{\mathcal{A}}}
\newcommand{\calP}{\ensuremath{\mathcal{P}}}

\newcommand{\ignore}[1]{}
\newcommand\questionEq{\stackrel{\mathclap{\normalfont\mbox{?}}}{=}}

\algnewcommand\algorithmicswitch{\textbf{switch}}
\algnewcommand\algorithmiccase{\textbf{case}}

\algdef{SE}[SWITCH]{Switch}{EndSwitch}[1]{\algorithmicswitch\ #1\ \algorithmicdo}{\algorithmicend\ \algorithmicswitch}%
\algdef{SE}[CASE]{Case}{EndCase}[1]{\algorithmiccase\ #1}{\algorithmicend\ \algorithmiccase}%
\algtext*{EndSwitch}%
\algtext*{EndCase}%
\newcommand{\CONGEST}{\ensuremath{\mathsf{CONGEST}}\xspace}
\newcommand{\LOCAL}{\ensuremath{\mathsf{LOCAL}}\xspace}

\newcommand{\eps}{\varepsilon}

\newcommand{\set}[1]{\left\{#1\right\}}

\DeclareMathOperator{\polylog}{polylog}
\DeclareMathOperator{\poly}{poly}

\newcommand{\complexityclass}[2][]{\ensuremath{\mathsf{#2}\ifthenelse{\isempty{#1}}{}{(#1)}}}

\newcommand{\hide}[1]{}

\newcommand{\FullOrShort}{short}

\ifthenelse{\equal{\FullOrShort}{full}}{

  \newcommand{\fullOnly}[1]{#1}
  \newcommand{\shortOnly}[1]{}

}{

  \newcommand{\shortOnly}[1]{#1}
  \newcommand{\fullOnly}[1]{}
}
  
\begin{document}

\title{On the Use of Randomness in Local Distributed Graph Algorithms}

%\begin{flushleft}

%\author{Anonymous Authors}
\author{\\
Mohsen Ghaffari \\
ETH Zurich, Switzerland \\
ghaffari@inf.ethz.ch 
\and
\\ 
Fabian Kuhn \\
University of Freiburg, Germany \\
kuhn@cs.uni-freiburg.de}

\date{}
\maketitle

\begin{abstract}
%\normalsize
  % One of the fundamental open problems in the area of distributed graph algorithms is to understand the role of randomness in local distributed graph algorithms. For many important problems, the best randomized algorithm have a time complexity that is polylogarithmic in the number of nodes $n$, whereas the best deterministic algorithms require at least $2^{\Omega(\sqrt{\log n})}$ time and are thus exponentially slower.
  \noindent We attempt to better understand randomization in local
  distributed graph algorithms by exploring how randomness is used and
  what we can gain from it:
  \begin{itemize}
  \item We first ask the question of how much randomness is needed to
    obtain efficient randomized algorithms. We show that for \emph{all}
    locally checkable problems for which $\poly\log n$-time
    randomized algorithms exist, there are such algorithms even if
    either (I) there is a only a single (private) independent random bit in each
    $\poly\log n$-neighborhood of the graph, (II) the (private) bits of randomness of different nodes are only $\poly\log n$-wise independent, or 
		(III) there are only $\poly\log n$ bits of global shared randomness (and no private
    randomness).
  \item Second, we study how much we can improve the error probability
    of randomized algorithms. For all locally checkable problems for
    which $\poly\log n$-time randomized algorithms exist, we show
    that there are such algorithms that succeed with probability
    $1-n^{-2^{\eps(\log\log n)^2}}$ and more generally
    $T$-round algorithms, for $T\geq \polylog\ n$, that succeed with probability
    $1-n^{-2^{\eps\log^2T}}$. We also show that
    $\poly\log n$-time randomized algorithms with success
    probability $1-2^{-2^{\log^\eps n}}$ for some $\eps>0$ can be
    derandomized to $\poly\log n$-time deterministic algorithms.
  \end{itemize}
  Both of the directions mentioned above, reducing the amount of randomness and
  improving the success probability, can be seen as partial derandomization of existing randomized algorithms. In all the above
  cases, we also show that any significant improvement of our results
  would lead to a major breakthrough, as it would imply significantly
  more efficient deterministic distributed algorithms for a wide class of problems.
\end{abstract}

\clearpage

%\vspace{-5pt}

\section{Introduction}
\label{sec:intro}
%\vspace{-5pt}

The gap between the complexity of randomized and deterministic distributed algorithms for local graph problems is one of the foundational, deep, and long-standing questions in distributed algorithms. A well-known special case is the question of Linial from 1987\cite{linial1987LOCAL, linial92} about the Maximal Independent Set (MIS) problem: While we have known randomized MIS algorithms that work in $O(\log n)$ rounds of the $\mathsf{LOCAL}$ model---i.e., synchronous message passing rounds---since the celebrated work of Luby\cite{luby86} and Alon, Babai, and Itai\cite{alon86}, Linial's question for obtaining a \emph{deterministic} algorithm that computes an MIS in $\poly(\log n)$ rounds still remains open. There is an abundance of similar questions about other concrete graph problems, several of which remain open. See e.g., the first five problems\footnote{Though, the last two of these are no longer open\cite{FischerGK17, ghaffari2018derandomizing}.} in the open problems chapter of the book by Barenboim and Elkin\cite{barenboimelkin_book}. More generally, we can ask whether
\begin{center}
$\mathsf{P}$-$\mathsf{LOCAL} \questionEq \mathsf{P}$-$\mathsf{RLOCAL}$
\end{center}
Here, $\mathsf{P}$-$\mathsf{LOCAL}$ denotes the family of \emph{locally checkable problems}\footnote{To make the question more widely applicable, we use a relaxed version of local checkability, where the local checking radius can be up to polylogarithmic in $n$. For a precise definition, we refer to \Cref{sec:model}.} that can be solved by deterministic algorithms in $\poly(\log n)$ rounds in $n$-node graphs and $\mathsf{P}$-$\mathsf{RLOCAL}$ denotes the family of \emph{locally checkable problems} that can be solved by randomized algorithms in $\poly(\log n)$ rounds, with success probability $1-1/n$. Both of these are with respect to the $\mathsf{LOCAL}$ model. One may view the above question as an analog of the well known $\mathsf{P}$ vs. $\mathsf{BPP}$ question in centralized computational complexity, i.e., (deterministic) polynomial-time  vs.\ bounded-error probabilistic polynomial-time~\cite[Chapter 7]{arora2009computational}. However, as we will see below, the questions are inherently very different.

%In this paper, we do \emph{not} make any progress on the above question. Instead, we take a step back and ask
In this paper, we try to shed more light on this fundamental question by taking a step back and asking
\emph{``what is the randomness used in these randomized distributed graph algorithms?"} That is, (A) how much randomness is needed, and (B) how strong are the probabilistic guarantees that randomized algorithms can provide? Both of these questions can be used as means for interpolating between %the two extremes of
 randomized and deterministic algorithms: Randomized algorithms, under the standard definition, can use an unbounded number of independent random bits at different nodes and they guarantee success with probability $1-1/n$. Deterministic algorithms use no randomness and they always guarantee success (which is at least as strong as guaranteeing success with probability $1$). 

Before delving into our answers to these questions, let us review some of the recent work centered on the gap between deterministic and randomized algorithms.

\subsection{An Overview of the Recent Developments on DET vs. RAND}

Over the past decade, there has been a number of beautiful developments, which are related to the aforementioned deterministic versus randomized question. We give a brief overview here.\footnote{This is certainly not exhaustive and it probably does not do justice to all the recent progress. We discuss only the cases that are most directly related to the current paper, in our understanding.}

\paragraph{Shattering Method and its Necessity:} One of the influential developments of the past decade in distributed algorithms for local graph problems was the introduction and wide usage of the shattering method. The method, which is inspired by Beck's algorithmic version of the Lov\'{a}sz Local Lemma\cite{beck1991algorithmic}, was first introduced in the distributed setting by Barenboim et al.~\cite{barenboim_symmbreaking}. In a rough sense, the method leads to randomized algorithms with two phases: a first efficient randomized phase that typically works in time that only depends on local graph parameters such as the maximum degree $\Delta$ and leaves a graph made of only small components, e.g., each of $\poly(\log n)$ size; and a second phase that solves each of these connected components separately, all in parallel, using deterministic algorithms. Hence, the dependency on the network size $n$ in randomized algorithms is brought down to the deterministic complexity for networks of size $N=\poly(\log n)$. For instance, for MIS, we know an $O(\log \Delta) + 2^{O(\sqrt{\log\log n})}$-round randomized algorithm with success probability $1-1/n$\cite{gmis}, and the second complexity term here mirrors the $2^{O(\sqrt{\log n})}$ complexity of the best known deterministic MIS algorithm~\cite{panconesi-srinivasan}. We run deterministic algorithms in the second phase, because the usual error probability bound of randomized algorithms, which is $1/\poly(N)$ for $N$-node instances, is not enough for a union bound over all components. Here, the deterministic vs.\ randomized question is more about the \emph{success probability}, rather than the bits of randomness.

More surprisingly, Chang et al.\cite{chang16} showed that the randomized complexity of any locally checkable problem on $n$-node networks (with success probability $1-1/n$) is at least its deterministic complexity on graphs with $\sqrt{\log n}$ nodes. Thus, if one improves the $n$-dependency of the randomized algorithms compared to the shattering-based results above, that improves also the deterministic complexity. This underlines the importance of understanding the complexity of deterministic algorithms, even if eventually we only care about randomized algorithms.

%It is worth noting that, to the best of our knowledge, such a relation showing that randomized algorithms cannot be faster than deterministic algorithms on some smaller instance is a unique feature of $\mathsf{LOCAL}$ graph algorithms. We are not aware of any such relation in any other computational model. 

\paragraph{Exponential Separations in the Landscape of Lower Complexities:}
Another significant recent development was the emergence of provable exponential separations between randomized and deterministic algorithms, though in a complexity regime below $O(\log n)$. Brandt et al.\cite{brandt} showed a lower bound of $\Omega(\log\log n)$ on the round complexity of computing a sinkless orientation in constant-degree graphs, which also implied a similar lower bound for $\Delta$-coloring trees of degree $\Delta=O(1)$. 
Chang et al.\cite{chang16} extended these to $\Omega(\log n)$ lower bounds for deterministic algorithms. They also exhibited an $O(\log\log n)$ round randomized algorithm and an $O(\log n)$ round deterministic algorithm for $\Delta$-coloring trees, hence proving that these complexities are tight and they have an exponential separation. Ghaffari and Su~\cite{GS17} later showed that the original problem of sinkless orientation (which is a special case of the Lov\'{a}sz Local Lemma) also exhibits the same exponential separation, by providing a $\Theta(\log\log n)$-round randomized and a $\Theta(\log n)$-round deterministic algorithm for it.
We emphasize that this exponential separation is between complexities that are in $O(\log n)$. To the best of our understanding, this separation has no bearing on the $\mathsf{P}$-$\mathsf{LOCAL}$ vs. $\mathsf{P}$-$\mathsf{RLOCAL}$ question or particular cases of it such as Linial's MIS question.

\paragraph{A Complexity-Theoretic Study and Derandomization:} As a step toward studying deterministic versus randomized complexities, 
Ghaffari et al.\cite{ghaffari2017complexity} introduced the sequential local model $\mathsf{SLOCAL}$: here a sequential algorithm processes vertices in an arbitrary order $v_1, v_2, \dots, v_n$, each time deciding about the output of the vertex $v_i$ that is being processed---e.g., its color in the coloring problem---by reading the current information available within a small $r$-hop neighborhood of $v_i$, including the related topology, and then recording the result (and potentially the gathered information) in the node $v_i$. See \cite{ghaffari2017complexity} for the precise definitions. The parameter $r$ is the locality of such an $\mathsf{SLOCAL}$ algorithm. The model provides a generalization of sequential greedy processes for problems such as MIS and $\Delta+1$ coloring, which both can be solved with locality $1$ in the $\mathsf{SLOCAL}$ model. The model may seem too powerful at first sight, given that it allows sequential processing, but it is no more powerful than adding randomization to the $\mathsf{LOCAL}$ model: It was shown in \cite{ghaffari2017complexity} that any (randomized or deterministic) $\mathsf{SLOCAL}$ algorithm with locality $\poly(\log n)$ for any locally checkable problem can be transformed to a $\poly(\log n)$-round randomized $\mathsf{LOCAL}$ algorithm. Given this, Ghaffari et al.\cite{ghaffari2017complexity} studied the question of $\mathsf{P}$-$\mathsf{SLOCAL}$ vs. $\mathsf{P}$-$\mathsf{LOCAL}$: whether any locally checkable problem that admits a deterministic $\mathsf{SLOCAL}$ algorithm with locality $\poly(\log n)$ can be solved using a $\poly(\log n)$-round deterministic $\mathsf{LOCAL}$ algorithm. A number of problems were shown to be complete with respect to $\mathsf{P}$-$\mathsf{SLOCAL}$\cite{ghaffari2017complexity,ghaffari2018derandomizing}, in the following sense: they admit deterministic $\mathsf{SLOCAL}$ algorithms with locality $\poly(\log n)$ and if one can provide a $\poly(\log n)$-round deterministic $\mathsf{LOCAL}$ for any of them, one has proven that $\mathsf{P}$-$\mathsf{SLOCAL} = \mathsf{P}$-$\mathsf{LOCAL}$. Example complete problems include network decompositions, splitting\cite{ghaffari2017complexity}, and certain locally verifiable versions of approximating dominating set or set cover\cite{ghaffari2018derandomizing}.

The work of \cite{ghaffari2017complexity} investigated \emph{a part} of the question of randomized vs. deterministic by examining $\mathsf{P}$-$\mathsf{SLOCAL}$ vs. $\mathsf{P}$-$\mathsf{LOCAL}$. A different work of Ghaffari et al.\cite{ghaffari2018derandomizing} showed that this actually \emph{fully} captures the issue, by proving that $\mathsf{P}$-$\mathsf{RLOCAL} = \mathsf{P}$-$\mathsf{SLOCAL}$. That is, any $\poly(\log n)$-round randomized $\mathsf{LOCAL}$ algorithm that solves a locally checkable problem with high probability can be derandomized into a deterministic $\mathsf{SLOCAL}$ algorithm for the same problem with locality $\poly(\log n)$. Hence, the question $\mathsf{P}$-$\mathsf{SLOCAL}$ vs. $\mathsf{P}$-$\mathsf{LOCAL}$ is equivalent to $\mathsf{P}$-$\mathsf{RLOCAL}$ vs. $\mathsf{P}$-$\mathsf{LOCAL}$ and the aforementioned problems are complete also with respect to $\mathsf{P}$-$\mathsf{RLOCAL}$. All of those problems admit $\poly(\log n)$-round randomized $\mathsf{LOCAL}$ algorithms, and any $\poly(\log n)$-round deterministic $\mathsf{LOCAL}$ algorithm for any of them would imply that $\mathsf{P}$-$\mathsf{RLOCAL} = \mathsf{P}$-$\mathsf{LOCAL}.$ Given what is known about these complete problems, and particularly network decomposition~\cite{panconesi-srinivasan}, the best known derandomization is that any $\poly(\log n)$-round randomized algorithm for \emph{any} locally checkable problem can be derandomized to a $2^{O(\sqrt{\log n})}$-round deterministic algorithm~\cite{ghaffari2018derandomizing}.

As a side remark, it is worth noting that the complexity-theoretic view mentioned above and some of the algorithmic and derandomization tools developed around it have already had concrete algorithmic applications: In particular, \cite{FischerGK17, ghaffari2018derandomizing} resolved a couple of the open questions regarding deterministic vs.\ randomized distributed algorithm for particular graph problems, including Open Problems 11.4, 11.5 and 11.10 of the book of Barenboim and Elkin\cite{barenboimelkin_book}. 

\subsection{Our Contribution}

As mentioned above, we try to shed more light on the $\mathsf{P}$-$\mathsf{RLOCAL}$ vs. $\mathsf{P}$-$\mathsf{LOCAL}$ problem by asking two questions about randomized distributed graph algorithms: (A) ``how much" randomness do they need for their efficiency, to solve in particular the above complete problems---e.g., network decompositions with $\poly(\log n)$ parameters---in $\poly(\log n)$-rounds? (B) And what kind of a bound can we prove on their error probability, given some limit on the time complexity. %We mill make these questions more formal soon. 
In each direction, we provide results that are in some sense the strongest that we one can achieve, barring a major breakthrough. More concretely, if one achieves a considerably stronger result, that would either prove $\mathsf{P}$-$\mathsf{SLOCAL} = \mathsf{P}$-$\mathsf{LOCAL}$ or at least provide a much faster deterministic algorithm for all problems in $\mathsf{P}$-$\mathsf{RLOCAL}$, including MIS, $(\Delta+1)$-coloring, and network decomposition with $\poly(\log n)$ parameters. 

Before diving to the answers, let us make something concrete. In much of the discussions in this paper, instead of talking about all possible randomized algorithms for all problems, we will focus directly on problems that are now known to be complete with respect to the $\mathsf{P}$-$\mathsf{RLOCAL}$ vs.\ $\mathsf{P}$-$\mathsf{LOCAL}$ question, in the sense mentioned above. In particular, much of our focus will be on network decompositions as introduced in \cite{awerbuch89}. A network decomposition of $G=(V, E)$ with $\alpha$ colors and diameter $\beta$---often abbreviated as an $(\alpha, \beta)$ network decomposition---is a partitionong of $V$ into $\alpha$ disjoint sets $V_1,\dots, V_\alpha$ such that for each $i\in\set{1,\dots,\alpha}$, each connected component of the induced subgraph $G[V_i]$ has diameter at most $\beta$ (see \Cref{sec:model} for a formal definition). In the following, when referring to a network decomposition with $\poly(\log n)$ parameters, we mean that $\alpha = \poly(\log n)$ and $\beta=\poly(\log n)$. It is known that every $n$-node network admits an $(O(\log n), O(\log n))$ network decomposition and that such a decomposition can be computed in $O(\log^2n)$ rounds of the \LOCAL model using randomized algorithms, with success probability $1-1/n$~\cite{linial93}. The best known deterministic algorithm to compute a network decomposition with $\poly(\log n)$ (or even weaker) parameters is exponentially slower and requires $2^{O(\sqrt{\log n})}$ rounds\cite{awerbuch96,panconesi-srinivasan}. We ask how much randomness is needed to compute network decompositions with $\poly(\log n)$ parameters (or some other related questions such as splitting) or what kind of a probabilistic guarantees can randomized algorithms provide in a given amount of time, for these problems. Given the aformentioned completeness results, we know that these concrete questions capture the role of randomness for all locally checkable problems: For instance, if we can construct a network decomposition with $\poly(\log n)$ parameters in $\poly(\log n)$ time using a certain ``amount of randomness'', then the same is true for any problem in $\mathsf{P}$-$\mathsf{RLOCAL}$ (using the same amount of randomness).
 
\subsubsection{Direction 1 --- How much randomness is needed?} 
We formalize the question about the amount of randomness, in three different ways: (A) the number of bits in each ``local neighborhood'', (B) the independence of the bits in different nodes, and (C) the number of bits shared in the whole network. We next discuss these cases, separately.

\smallskip
\noindent (\textbf{A}) First, we note that the standard definition for randomized algorithms allows each node to have unbounded amount of randomness. Usual algorithms need less than this and use only $\poly(\log n)$ bits per node~\cite{linial93}. We show in \Cref{thm:onebit-perh-hop1,thm:onebit-perh-hop2} that even much less than that suffices: even if for each node there is just one bit of randomness somewhere within its $\poly(\log n)$ hops, and these bits are independent of each other, then we can still compute all $\mathsf{P}$-$\mathsf{RLOCAL}$ problems in $\poly(\log n)$. We also give an algorithm that, in such a setting, builds a network decomposition with $\poly(\log n)$ parameters, in $\poly(\log n)$ rounds of the $\mathsf{CONGEST}$ model, with probability $1-1/\poly(n)$. Recall that the $\mathsf{CONGEST}$ model is a variant of the $\mathsf{LOCAL}$ model where the message sizes are limited to $O(\log n)$ bits, in contrast to the $\mathsf{LOCAL}$ model which allows unbounded message sizes.  

Two comments are in order. First, this one bit per $\poly(\log n)$-hop neighborhood is in some sense the least that we need to assume. Otherwise, there can be large $\poly(\log n)$-hop neighborhoods where there is no randomness, and hence $\poly(\log n)$-time algorithms that work in those areas are actually providing a deterministic algorithm. Second, our theoretical result might have some message for practical settings when thinking outside the realm of the worst-case analysis. Probably, in many networking settings, it is reasonable to assume that one can ``extract'' one bit of randomness out of various network properties (topology, identifiers, etc.) in each $\poly(\log n)$ neighborhood. In a very informal sense, our result implies that there are probably reasonable ways of building efficient ``deterministic'' algorithms (with only $1$ bit of pseudo-randomness per $\poly(\log n)$-hop neighborhood), such that breaking these algorithms requires very carefully built input networks\footnote{There is much to be investigated here. What kind of randomness extraction is possible, under different network assumptions? Also, can we formalize the sufficiency of such randomness against a computationally bounded adversary who builds the input graph, perhaps using some cryptographic assumptions?}.

\smallskip
\noindent (\textbf{B}) Second, we investigate the independence between the random bits. Standard algorithms assume the bits of different nodes to be fully independent. In \Cref{thm:NetDecomp-with-K-wise}, we show that $\poly(\log n)$-wise independence suffices for computing a network decomposition with $\poly(\log n)$ parameters in $\poly(\log n)$ rounds of the $\mathsf{LOCAL}$ model. Hence, any problem in $\mathsf{P}$-$\mathsf{RLOCAL}$ can be solved in $\poly(\log n)$ rounds, even if the random bits of different nodes are $\poly(\log n)$-wise independent.

\smallskip
\noindent (\textbf{C}) Third, as a result of (B), we can also bound the total number of random bits required in the whole network. Standard models of randomized algorithms implicitly assume a total of at least $\Omega(n)$ bits, i.e., at least one bit per node. We show that just $\poly(\log n)$ shared random bits suffice. For the local splitting problem shown to be complete in \cite{ghaffari2017complexity}, we prove in \Cref{lem:splitting} that just $O(\log n)$ bits of shared randomness suffices. This concretely shows that the $\mathsf{P}$-$\mathsf{RLOCAL}$ vs.\ $\mathsf{P}$-$\mathsf{LOCAL}$ question is very different from its well-known centralized $\mathsf{P}$ vs. $\mathsf{BPP}$ analog (polynomial-time  vs.\ bounded-error probabilistic polynomial time)\cite{sipser2006introduction, arora2009computational}. Note that for this centralized question, any probabilistic algorithm with $O(\log n)$ bits of randomness can be derandomized trivially in polynomial time, by checking all the $2^{O(\log n)} = n^{O(1)}$ possibilities for the $O(\log n)$ bits of randomness. Our result about splitting shows that even problems that are solvable efficiently with $O(\log n)$ bits of randomness can be hard with respect to $\mathsf{P}$-$\mathsf{RLOCAL}$ vs.\ $\mathsf{P}$-$\mathsf{LOCAL}$. Furthermore, we note that this bound of $O(\log n)$ bits is asymptotically the least that we can assume, unless we come up with a deterministic algorithm\footnote{An algorithm with $b< \log n$ bits of randomness is a uniformly random choice among $2^b< n$ possible deterministic algorithms. Each deterministic algorithm either succeeds or fails. Hence, the highest probability that is not $1$ would be at most $1-1/2^{b} < 1-1/n$. Thus, if the probability is at least $1-1/n$, it is equal to $1$.}.

Perhaps as more interesting end results, we also show that something similar can be said about more standard graph problems, for instance network decomposition. The result of the item (B) discussed above, combined with standard constructions for $k$-wise independent bits\cite{alon2004probabilistic}, shows that $\poly(\log n)$ bits of shared randomness build network decompositions with $\poly(\log n)$ parameters, in $\poly(\log n)$ rounds of the $\mathsf{LOCAL}$ model. We show such a result also for the much more stringent $\mathsf{CONGEST}$ model, using a very different method: In \Cref{thm:NetDecomp-with-Shared-Rand}, we show that we can build network decompositions with $\poly(\log n)$ parameters, in $\poly(\log n)$ rounds of the $\mathsf{CONGEST}$ model, using only $\poly(\log n)$ bits of shared randomness (and no private randomness). 

%\smallskip

\subsubsection{Direction 2 --- Error probability vs Round Complexity}

Usually, the study of distributed graph algorithms has focused only on two regimes of error probabilities: an error probability of $0$, for deterministic algorithms, and an error probability of $1/\poly(n)$, for randomized algorithms. In a number of places where we need deterministic algorithms, the main property that we require from the algorithms is that the error probability is $0$ or close to $0$ and not how many random bits we use. Our second direction is to explore the trade-off between error probability and running time. Concretely, how small can we make the error probability in a certain time budget? Moreover, at what point is the error probability small enough so that we can derandomize the algorithm completely? 

For any $T\geq \log^{c} n$, where $c$ is a sufficiently large constant, we give a randomized algorithm that succeeds with probability at least $1-2^{-2^{\eps\log^2 T}}=1-n^{-2^{\eps\log^2 T}}$ for some constant $\eps>0$ and that computes a network decomposition with cluster diameter $T$ and $T$ cluster colors in $T$ rounds of the \CONGEST model (see \Cref{thm:ep_alg}). For instance, for $T=\polylog n$, this results in an error probability of $n^{-2^{\eps\log^2 \log n}}$, which is much stronger than the $n^{-\Theta(1)}$ error probability bound of standard randomized algorithms. In \Cref{thm:ep_derandomization}, we also show that the result of \Cref{thm:ep_alg} is nearly the best error probability that one can guarantee, unless we improve the deterministic complexity of network decompositions. Concretely, any randomized algorithm with round complexity $T$---even in the $\mathsf{LOCAL}$-model---that has success probability at least $1-2^{2^{\eps \log^\beta} T}$ for any constant $\eps>0$ and any constant $\beta>2$ would imply a deterministic network decomposition with $\poly(\log n)$ parameters in $2^{O(\log^{1/\beta} n)} \ll 2^{O(\sqrt{\log n})}$ rounds, thus significantly improving on the long-standing bounds of Panconesi and Srinivasan\cite{panconesi-srinivasan}. Similarly, in \Cref{thm:polylogderand}, we show that any randomized $\poly(\log n)$-round algorithm with success probability better than $1-2^{-2^{\log^\eps n}}$, for an arbitrary constant $\eps>0$, would imply a $\poly(\log n)$-round network decomposition with $\poly(\log n)$ parameters, thus proving that $\mathsf{P}$-$\mathsf{RLOCAL} = \mathsf{P}$-$\mathsf{LOCAL}.$ Previously, the best known such derandomization result was that any randomized algorithms with a much stronger success probability of $1-2^{-\Theta(n^2)}$ can be derandomized (via a union bound over all $2^{\Theta(n^2)}$ many possibilities for $n$-node graphs). This was implicit in Theorem 3 of \cite{chang16}.

%\newpage

%%% Local Variables:
%%% mode: latex
%%% TeX-master: "main"
%%% End:

\section{Model and Preliminaries}
\label{sec:model}

\noindent\textbf{Communication Model:} We work with two closely
related models of distributed computing, \LOCAL and \CONGEST: The
communication network is abstracted as an $n$-node graph $G=(V, E)$,
with one processor on each node $v\in V$ which has a unique
identifier. We typically assume that the identifiers are represented
by $\Theta(\log n)$ bits. Communication happens in synchronous rounds, where per round each node can send one message to each neighbor. In the \LOCAL model, message sizes can be unbounded. In the \CONGEST model, each message can have $O(\log n)$ bits. At the beginning, the processors/nodes do not know the topology of the network, except for potentially knowing some global parameters (as we shall discuss next). At the end, each processor should know its own part of the output, e.g., its color in the vertex coloring problem.
 
\paragraph{Uniform and Non-Uniform Algorithms:}
In most cases, we assume that the nodes of a distributed algorithm
initially know the number of nodes $n$ or an upper bound on $n$.x We
call such an algorithm a \emph{non-uniform} distributed algorithm and
we call an algorithm where the nodes initially do not know anything
about $n$ a \emph{uniform} algorithm. We formally model the knowledge
of $n$ as follows. In a non-uniform distributed algorithm, all nodes
are given $n$ as input. We say that a non-uniform algorithm $\calA$
solves a distributed graph problem $\calP$ in time $T(n)$ if $\calA$
solves $\calP$ in time at most $T(n)$ on all graph with at most $n$
nodes. If the correctness of a solution to a graph problem $\calP$
depends on the number of nodes $n$, we use the notation $\calP(n)$ to
make this clear. A solution to problem $\calP(n)$ that satisfies the
requirements for graphs with at most $n$ nodes. If we for example
compute an $(O(\log n), O(\log n))$-decomposition of a graph $G=(V,E)$
with at most $n$ nodes, the cluster diameter and number of cluster
colors can depend logarithmically on $n$ rather than just on the
actual number of nodes $|V|$. The error probability of a non-uniform
algorithm is defined as follows.

\begin{definition}[Error Probability]\label{def:errorprobability}
  We say that a (non-uniform) distributed algorithm $\calA$ solves a
  given distributed graph problem $\calP$ on $n$-node graphs with
  error probability $\delta(n)$ in time $T(n)$ if the following holds. When given $n$ as
  an input, $\calA$ computes a correct solution to $\calP$ with
  probability at least $1-\delta(n)$ on all graphs $G$ with at most $n$
  nodes.
\end{definition}

The definition implies that a randomized algorithm has to succeed with probability $1-\delta(n)$ even if the actual
graph has fewer than $n$ nodes. If $\delta(n)\leq 1/n^c$ for a constant
$c>1$ that can be chosen sufficiently large, we say that algorithm $\calA$ solves the problem $\calP$ with high
probability (w.h.p.).

\paragraph{Local Checkability:}
As mentioned in the introduction, we study distributed graph problems
where the validity of a solution is locally checkable
\cite{fraigniaud13}. Roughly, a graph problem $\calP$ is said to be
$d$-locally checkable if given a solution to $\calP$, there exists a
deterministic $d$-round algorithm $\calA_C$ in the \LOCAL model such
that every node outputs ``yes'' if and only if the given solution is a
valid solution to $\calP$. The class of problems we consider contains
the well-known class of locally checkable labeling (LCL)
problems\footnote{LCL problems are graph problems where the output of
  each node is a label from a constant-size alphabet and where the
  correctness of a solution can be checked with a constant-time \LOCAL
  algorithm \cite{naor95}.}, however, we use a much looser definition
of local checkability, where in particular the checking radius can
depend on $n$ and which is formally defined as follows.

\begin{definition}[Local Checkability]\label{def:localcheckability}
  Let $\calP$ be a distributed graph problem. We say that $\calP$ is
  \emph{$d(n)$-locally checkable} for a function $d(n)$ if there
  exists a deterministic non-uniform distributed algorithm $\calA_C$,
  which is given $n$ as input an which has the following properties.
  Given a graph $G=(V,E)$ of size $|V|\leq n$ and values $x_v$ for
  $v\in V$, $\calA_C$ has round complexity at most $d(n)$ and it
  checks whether $\set{x_v : v\in V}$ is a correct solution for
  $\calP$. After running $\calA_C$, each node outputs ``yes'' or
  ``no'' such that all nodes output ``yes'' if and only if
  $\set{x_v : v\in V}$ is a correct solution for $\calP$. We say that
  $\calP$ is \emph{strictly $d(n)$-locally checkable} if the round
  complexity is $d(|V|)$, i.e., if the checking radius only depends on
  the actual number of nodes and not on the upper bound $n$.
\end{definition}

Note that any LCL problem and more generally any problem that is
$d$-locally checkable for a constant $d\geq 0$ is also strictly
locally checkable.

\paragraph{Network Decomposition:} As discussed, the complexity of
computing a network decomposition is at the core of understanding the
role of randomization in local distributed graph algorithms: Given a
network decomposition with $\poly(\log n)$ parameters for a sufficiently
large (polylogarithmic) power $G^r$ of the network graph $G$, any randomized
$\poly(\log n)$-time algorithm for a $\poly(\log n$)-locally checkable problem can be
derandomized  to a deterministic $\poly(\log n)$-time  algorithm \cite{ghaffari2017complexity,ghaffari2018derandomizing}.
We slightly adapt the definition of a network decomposition from the one
introduced in \cite{awerbuch89,linial93} to make it more directly
useful in the \CONGEST model. Given a graph $G=(V,E)$, a strong
$(d(n), c(n))$-decomposition of $G$ is partition of $V$ into clusters
$C_1,\dots,C_p$ together with a subtree $T_i$ of $G$ and a color
$\gamma_i\in\set{1,\dots,c(n)}$ for each cluster $C_i$. The tree $T_i$
of cluster $C_i$ contains all nodes of $C_i$ (i.e., $T_i$ spans the
cluster $C_i$). Each tree $T_i$ has diameter  at most
$d(n)$ (which implies that each cluster has weak diameter at most
$d(n)$) and the colors of the clusters are chosen such that clusters
that are connected by an edge of $G$ are assigned different colors. We
say that a decomposition has \emph{congestion $\kappa\geq 1$} if each node is
containted in at most $\kappa$ clusters of each color.\footnote{For
  efficient usage in \CONGEST, one could require
  that each edge is only used by a few clusters per color.}
Usually, in the literature, the trees spanning the clusters are not
given as part of the definition and instead of explicitly specifying
the congestion of a decomposition, the literature distinguishes
between strong and weak diameter decompositions. The decompositions of
the above definition have weak diameter $d(n)$. In a strong diameter
decomposition, the tree $T_i$ of each cluster $C_i$ consists exactly
of the nodes in $C_i$ such that each node participates only in the
tree of its cluster. A strong diameter decomposition is therefore a
special case of a decompositions with congestion $1$.

%%% Local Variables:
%%% mode: latex
%%% TeX-master: "main"
%%% End:

%\newpage
\section{Viewing Randomness as a Scarce Resource}
\label{sec:randomness}
In this section, we view randomness as a resource and use this perspective to interpolate between standard randomized algorithms and deterministic algorithms. Deterministic distributed algorithms do not use any randomness. Randomized distributed algorithms, in the standard definitions, can use an unbounded number of bits of randomness in each node of the network, where all the bits (in the same node and also across the whole network) are assumed to be independent of each other. As mentioned before, for many of the classic problems in distributed graph algorithms, the known randomized algorithms are considerably more time-efficient compared to their deterministic counterparts. Our goal in this section is to investigate this gap, by viewing the bits of randomness as a scarce resource and asking ``how much" of it is really needed for ``efficiency". We next discuss how we make these two phrases of how much and efficiency more concrete.

Regarding ``efficiency", our concrete objective is to be able to solve the classic local problems in distributed graph algorithms (e.g., network decompositions, maximal independent set, and $\Delta+1$ coloring) in time polylogarithmic in $n$, given the limited randomness that we have. Recall that this polylogarithmic time is usually construed as a first-order definition of efficiency~\cite{barenboim12_decomp} and it is achievable using randomized algorithms with unbounded randomness per node. Regarding ``how much" randomness, we formalize the question and study it in three different ways:

\begin{itemize}
\item[(\textbf{A})] What if instead of each node having access to its own (unbounded) source of random bits, we only have some few bits of randomness in the whole network but such that each node can reach at least one of them? Concretely, we assume that some nodes $S\subseteq V$ of the network hold some bits of randomness (which are independent of each other), each holding just a single bit, and for each node $v\in V$, there is at least one node $s\in S$ within distance $h$ hops of $v$. In order so that just accessing this bit of randomness is time-efficient, we will assume that $h=\poly(\log n)$. 

\item[(\textbf{B})] What if the bits of the randomness in the network are correlated and have only some limited independence, e.g., they that $k$-wise independent for $k=\poly(\log n)$?

\item[(\textbf{C})] What if instead of each node having its own independent bits, which amounts to $\Omega(n)$ bits in the whole network, we have only $k$ bits of global shared randomness? 
\end{itemize} 
We discuss these directions separately, (A) in \Cref{subsec:BitGathering}  and (B) and (C) in \Cref{subsec:sharedRand}.

\subsection{One Bit of Private Randomness Per {\boldmath$\poly(\log n)$} Hops}
\label{subsec:BitGathering}
In this subsection, we work under the assumption that there is one bit of randomness within $\poly(\log n)$ distance of each node, and we show that this amount suffices for building network decompositions with $\poly(\log n)$ colors and radius. Therefore, this amount of randomness suffices to have a $\poly(\log n)$-round $\mathsf{LOCAL}$ model algorithm for any (locally checkable) problem that can be solved in $\poly(\log n)$ rounds of the $\mathsf{LOCAL}$ model using unbounded randomness. We provide a construction of network decomposition in the $\mathsf{CONGEST}$ model, with the hope that this construction can be useful also in settings where we care about having small messages. 

\begin{theorem}\label{thm:onebit-perh-hop1}
Suppose that nodes $S\subseteq V$ of the network hold some independent bits of randomness, each holding just a single bit, and for each node $v\in V$, there is at least one node $s\in S$ within distance $h$ hops of $v$, where $h=\poly(\log n)$. Then, there is a $\poly(\log n)$-time distributed algorithm in the $\mathsf{CONGEST}$ model that, using only these bits as its source of randomness, constructs a $(O(\log n), h\poly(\log n))$-network decomposition of the the graph with congestion $1$.
\end{theorem}
\begin{proof}
  The proof of \Cref{thm:onebit-perh-hop1} consists of two parts, which we next present as \Cref{lem:BitGathering,lem:CONGEST-MPX-BallCarving}. In \Cref{lem:BitGathering}, we use a certain ruling set construction to cluster nodes into low-diameter clusters, each each cluster has some $\poly(\log n)$ bits of randomness, unless the cluster is a connected component on its own. Then, in \Cref{lem:CONGEST-MPX-BallCarving}, we use this randomness to build the desired network decomposition. \Cref{thm:onebit-perh-hop1} directly follows from the statements of \Cref{lem:BitGathering,lem:CONGEST-MPX-BallCarving}.
\end{proof}

Note that an undesirable property of \Cref{thm:onebit-perh-hop1} is that $h$ appears in the diameter of the network decomposition. In \Cref{thm:onebit-perh-hop2}, we explain how to fix that issue, using some of the methods that we will discuss in the next subsection. Before stating the two lemmas to prove \Cref{thm:onebit-perh-hop1}, we first recall the notion of ruling sets and the known deterministic algorithms for them.

\paragraph{Ruling Sets:} We need the notion of ruling sets as
introduced in \cite{awerbuch89} in some of our algorithms. Given a
graph $G=(V,E)$, a subset $U\subseteq V$ of the nodes of $G$, and two
parameters $\alpha,\beta \geq 1$, a $(\alpha,\beta)$-ruling set of $G$
w.r.t.\ $U$ is a subset $S\subseteq U$ of the nodes in $U$ such that
for all $x,y\in S$, $d_G(x,y)\geq \alpha$ and for all $x\in U$, there
exists a $y\in S$ such that $d_G(x,y)\leq \beta$. For any
$\alpha\geq 2$, a $(\alpha,\alpha\log n)$-ruling set of $G$ w.r.t.\
$S$ can be computed deterministically in time $O(\alpha\cdot\log n)$
in time $O(\alpha\cdot \log n)$ in the \CONGEST model
\cite{awerbuch89,CONGEST_rulingsets}.

Now, we are ready to provide the lemmas that prove \Cref{thm:onebit-perh-hop1}.

\begin{lemma}\label[lemma]{lem:BitGathering} Suppose that nodes $S\subseteq V$ of the network hold some independent bits of randomness, each holding just a single bit, and for each node $v\in V$, there is at least one node $s\in S$ within distance $h$ hops of $v$, where $h=\poly(\log n)$. Then, there is an $O(hk\log n)$-round deterministic algorithm in the $\mathsf{CONGEST}$ model that partitions the nodes into disjoint clusters, each inducing a connected subgraph with diameter $O(kh\log n)$, with the following property: Each cluster is either (A) isolated --- meaning that it has no neighboring cluster --- or (B) its center holds $k$ bits of randomness, which are independent of each other and independent of the bits held by other cluster centers.
\end{lemma}
\begin{proof}[\textbf{Proof of \Cref{lem:BitGathering}}]
First, we compute a certain $(h', h'\log n)$-ruling sets $R$, for $h'=10kh$. That is, any two nodes of $R$ have distance at least $h'$ from each other and moreover, for each node $v\in V$, there is at least one node of $R$ within its distance $h'\log n$. This can be computed directly in the $\mathsf{CONGEST}$ using the ruling set algorithm of Awerbuch et al.\cite{awerbuch89}, as we remarked in \Cref{sec:model}. Now, define clusters in the graph, one centered at each node in $R$, where each node $v\in V$ joins the cluster of the nearest $R$ node. Notice that this is doable in $h'\log n$ rounds, using a simple flooding of the name of nodes in $R$, where only the first name is propagated. Then, each node $v\in V$ knows its cluster center in $R$. Let us focus on one cluster $\mathcal{C}$ centered at node $r\in R$. There are two possibilities:

In the easier case of \emph{singularity} where the cluster has no neighboring cluster, we leave these clusters on their own; these satisfy property (A) in the lemma statement, and later, when it comes to building a network decomposition, we will color these clusters easily with color $1$, as one of the clusters of our network decomposition. 

In the less trivial case, suppose that $\mathcal{C}$ has at least one neighboring cluster $\mathcal{C}'$ centered at $r'\in R$. Then, on the shortest path $\mathcal{P}$ connecting $r$ to $r'$, the first $h'/3$ nodes $F\subset \mathcal{P}$ belong to $\mathcal{C}$ (they cannot belong to any other cluster, as that would be in contradiction with $R$ being $h'$ independent). In fact, for each node $w\in F$, even the whole $h'/6$-hop neighborhood of $w$ belongs to $\mathcal{C}$, for the same reason. Now, we can choose $h'/(10h)> k$ nodes  $F'\subset F$ such that any two nodes of $F'$ have distance at least $3h$ from each other. For each node $w\in F'$, there is some node $s_w\in S$ within $h$ hops of $w$ that holds a random bit. Moreover, these source nods $s_w$ are distinct for any different nodes $w, w' \in F'$, because any two nodes of $F'$ have distance at least $3h$. Since $h'/6$-hop neighborhood of any node in $F$ is in $\mathcal{C}$, and given that $h\gg h'/6$, all nodes $s_w$ for all $w\in F'$ are also in $\mathcal{C}$. Therefore, cluster $\mathcal{C}$ contains at least $h'/(10h) > k$ bits of randomness. We can propagate this amount of randomness to the cluster center $r$, by a simple upcast on the tree connecting $r$ to the nodes of $\mathcal{C}$, in $O(h'\log n+k) = O(h'\log n) = O(kh\log n)$ rounds.
\end{proof}

\begin{lemma}
\label{lem:CONGEST-MPX-BallCarving}Suppose that we are given a partitioning of the nodes into disjoint clusters, each inducing a connected subgraph with diameter $O(h\log^3 n)$, with the following property: Each cluster is either (A) isolated --- meaning that it has no neighboring cluster --- or (B) its center holds $C \log^2 n$ bits of randomness, which are independent of each other and independent of the bits held by other cluster centers. Then, there is a $h\cdot\poly(\log n)$-time distributed algorithm in the $\mathsf{CONGEST}$ model that, using only these bits as its source of randomness, constructs a $(O(\log n), h\poly(\log n))$-network decomposition of the the graph with congestion $1$.
\end{lemma}  

\begin{proof}[\textbf{Proof of \Cref{lem:CONGEST-MPX-BallCarving}}]
On a high-level, our network decomposition is obtained by running the randomized algorithm of Elkin and Neiman\cite{elkin16_decomp}, itself inspired by Blelloch et al.\cite{blelloch14} and Miller et al.\cite{miller2013parallel},  on top a logical network where we virtually contract each cluster to be a single node. This logical \emph{cluster graph} $\mathcal{CG}$ is obtained by viewing each cluster as one node and connecting two clusters if they include nodes that are adjacent, in the base graph. Each round of communication between two neighboring clusters can be performed in $O(h'\log n)$ rounds, given that the centers are within $O(h'\log n)$ hops of each other. One has to be careful with one subtlety: we cannot simulate a full-fledged $\mathsf{CONGEST}$ model on this cluster graph, e.g., if a cluster has to receive many different messages from many different neighboring clusters, we might not be able to deliver all of these messages to the center of the cluster, in time proportional to the cluster radius. However, in order to run the construction of Elkin and Neiman\cite{elkin16_decomp}, it will suffice for us to deliver an aggregate function of the messages sent by neighboring clusters to the cluster center, e.g., the minimum value in these messages. 

\paragraph{Construction:} We next give a brief overview of the algorithm of \cite{elkin16_decomp}. Our description is phrased as running on top of our clusters. The construction has $10\log n$ phases, where in each phase $i$ we gradually color some non-adjacent set of the clusters with color $i$ and remove them from the graph. Let each center $v$ of a cluster $\mathcal{C}$ pick a random value $r_v$ from a geometric distribution\footnote{Elkin and Neiman\cite{elkin16_decomp} wrote their description by choosing random variables $r_v$ from an exponential distribution, which can assume continuous values. We would like to explicitly talk about the number of random bits and for that, the geometric distribution is more convenient. The arguments of Elkin and Neiman\cite{elkin16_decomp} about the exponential distributed extends to its discrete analog, the geometric distribution; the core property of being a memoryless distribution which holds for both distributions.}, where $Pr[r_v=k] = 2^{-k}$. Then, each cluster $\mathcal{C}'$ centered at node $u$ considers the maximum two clusters according to the measure $r_v-dist_{\mathcal{CG}}(v, u)$, where $dist_{\mathcal{CG}}(v, u)$ is the distance between the clusters centered at $v$ and $u$ in the cluster graph $\mathcal{CG}$. Let $m_1$ and $m_2$ be the two maximum measures. If $m_1-m_2 > 1$, then the cluster $\mathcal{C}'$ centered at $u$ gets colored in this phase with color $i$. If $m_1-m_2 \in\{0,1\}$, then cluster $\mathcal{C}'$ remains for the next phase. Each phase runs in $O(\log n)$ rounds on top of the cluster graph $\mathcal{CG}$. This can be performed in $O(h\log^2 n)$ rounds on the base graph $G$, because each cluster $\mathcal{C}'$ centered at a node $u$ needs to pass to each of its neighbors only the top two cluster names $\mathcal{C}$ centered at node $v$ and radii $r_v$ according to measure $r_v-dist_{\mathcal{CG}}(v, u)$. Over all the $10\log n$ phases, this translates to a round complexity of $O(h\log^3 n)$.

\paragraph{Randomness:} We argue that the randomness that we have in clusters suffices for the construction. To choose the random value $r_v$ for each phase, having $10\log n$ bits suffices, with high probability: To generate $r_v$, think about the process of flipping coins one by one until the first tail coin shows up. The iteration of the first tail is the value $r_v$. With probability $1-1/n^{10}$, we toss at most $10\log n$ coins, before the first tail. Hence, to run all the $10\log n$ phases, $100\log^2 n$ bits suffice\footnote{In fact, $O(\log n)$ bits suffice for all phases as the number of coins until one sees $\Theta(\log n)$ tails is, w.h.p., $O(\log^2 n)$.}.

\paragraph{Properties of network decomposition:} We give only a sketch; the proof details can be found in \cite{elkin16_decomp}. First, as another corollary of the above, we also see that with high probability, we have $\max \{r_v\} = O(\log n)$. Thus, each cluster $\mathcal{C}'$ that gets color $i$ is at most $O(\log n)$ hops away from its central cluster $\mathcal{C}$, in the cluster graph $\mathcal{CG}$. Moreover, as argued in \cite[Lemma 4]{elkin16_decomp}, due the way of tie breaking by measure  $r_v-dist_{\mathcal{CG}}(v, u)$, we can see that (I) any two neighboring clusters that get color $i$ must have the same central cluster $\mathcal{C}$, and (2) all clusters on the shortest path in $\mathcal{CG}$ from $\mathcal{C}$ to $\mathcal{C}'$ are also colored with color $i$. Hence, they induce a connected subgraph with radius $O(\log n)$ in the cluster graph $\mathcal{CG}$ and thus also a connected subgraph with radius $O(h\log^3 n)$ in the base graph $G$. Finally, the probability of a cluster remaining uncolored in one phase is $1/2$\cite[Claim 6]{elkin16_decomp}. Hence, after $10\log n$ phases, with high probability, all clusters are colored. Hence, we get a strong-diameter network decomposition with $10\log n$ colors and diameter $O(h\log^3 n)$, i.e., an $((10\log n), O(h\log^3 n))$-network decomposition with congestion $1$.
\end{proof}

\subsection{Shared Randomness, and Private Randomness with Limited Independece}
\label{subsec:sharedRand}
We now ask how many bits of globally shared randomness are sufficient for efficiency (when there is no private randomness). Another way of viewing the question is asking how much randomness is needed, in total, over the whole network. Standard randomized algorithms use $\Omega(n)$ bits, e.g., at least one bit per node. The arguments of the previous section can be used to lower this somewhat but it might still be $\tilde{\Omega}(n)$ bits in some networks. But in fact, much less suffices, merely $\poly(\log n)$ bits. Incidentally, when showing this, we will also prove that in the standard model where each node has some private randomness, say e.g., $\poly(\log n)$ bits, we do not need these bits to be fully independent and it suffices if the bits in the whole network are only $\poly(\log n)$-wise independent. 

First, we investigate these questions in the $\mathsf{LOCAL}$ model. We explain that with just $O(\log n)$ bits of shared randomness, we can solve the splitting problem introduced by Ghaffari et al.\cite{ghaffari2017complexity}. This is a problem that nicely captures the power of randomness as it can be solved using randomized algorithms in zero rounds, and it was shown\cite{ghaffari2017complexity} that if one can solve it in $\poly(\log n)$ rounds deterministically, then we can derandomize all $\poly(\log n)$-round randomized algorithms for any locally checkable problem. In particular, \cite{ghaffari2017complexity} gives a reduction that solves network decomposition using $\poly(\log n)$ iterations of splitting, and this implies $\poly(\log n)$ bits of shared randomness suffice for network decomposition, in the $\mathsf{LOCAL}$ model. 
Then, we also show that similar ideas can be used to prove that $\poly(\log n)$-wise independence among the private bits of randomness in the network suffices for network decomposition in the $\mathsf{CONGEST}$ model. Finally, we turn our attention to the $\mathsf{CONGEST}$ model and show a more explicit algorithm (instead of reductions) that builds network decompositions in $\poly(\log n)$ rounds using  $\poly(\log n)$ bits of shared randomness.

\paragraph{Splitting in zero rounds, using \boldmath$O(\log n)$ bits of shared randomness:} Ghaffari et al.\cite{ghaffari2017complexity} defined a certain problem called \emph{splitting}, which can be solved using randomized algorithms in zero rounds with high probability, and showed that if one can solve this problem in $\poly(\log n)$ rounds deterministically, then one can derandomize all $\poly(\log n)$-round randomized algorithms for any locally checkable problem. We show that $O(\log n)$ bits of shared randomness suffice for solving this problem in zero rounds, with high probability. In the splitting problem, we are given a bipartite graph $H=(U, V, E)$ where each node in $U$ has at least $\Omega(\log^{c} n)$ neighbors in $V$ and we should color each node of $V$ red or blue so that each node of $U$ has at least one neighbor in each color. Here, $c$ can be set to be a desirably large constant $c\geq 1$. 

\begin{lemma}\label{lem:splitting} There is a randomized algorithm that using $O(\log n)$ bits of shared randomness solves the splitting problem, in zero rounds, with probability at least $1-1/n$.
\end{lemma}
\begin{proof}
Notice that coloring each node of $U$ randomly red or blue satisfies this constraint, with high probability (by applying a Chernoff bound on the neighborhood of each node in $U$ and then a union bound over all nodes in $U$). 

There are two well-known arguments for showing that a small amount of shared randomness suffices for this problem. First, thanks to the variant of Chernoff that holds for $p$-wise independent random variables \cite{schmidt1995chernoff}, we can see that $O(\log n)$-wise independence suffices for this argument. Moreover, we can build $\poly(n)$ bits that are $p$-wise independent, by using merely $O(p\log n)$ bits that are fully independent, using standard constructions, see e.g., \cite{alon2004probabilistic}. This means $O(p\log n) = O(\log^2 n)$ bits of shared randomness suffice for splitting. Second, a result of Naor and Naor\cite{naor1993small} can be used to show that even $O(\log n)$ bits suffice. They give a construction of $p$-wise $\eps$-bias spaces. Very roughly speaking, these are spaces that are approximately $k$-wise independent; see their paper for the definition, and also for the concentration inequalities that can be derived for such random variables. In their section 6.1, they show that a sample space of size $n^{O(1)}$ suffices for getting a coloring where each node has at least one neighbors in each color\footnote{Their phrasing is different. They talk about a problem called set balancing, in the context of discrepancy theory. But the problems are the same and when we assume that the minimum degree is at least $d$, their construction in section 6.1 suffices to get a discrepancy of $O(d^{1/2+\delta} \cdot \sqrt{\log n})$ between the two colors, for any constant $\delta>0$. Setting $\delta=0.1$ and $c\geq 1+4\delta=1.4$ in the definition of the splitting problem ensures that $d^{1/2+\delta} \cdot \sqrt{\log n} = o(d)$.}. This $n^{O(1)}$ size space means $O(\log n)$ bits of shared randomness suffice to sample from this space.  
\end{proof}

%\bigskip 
\paragraph{Network decomposition in \boldmath$\mathsf{LOCAL}$ using \boldmath$\poly(\log n)$-wise independent bits:}
There are known randomized construction of network decomposition with $\poly(\log n)$ parameters in $\poly(\log n)$ rounds of the $\mathsf{LOCAL}$ model\cite{linial93, elkin16_decomp}. These all assume that the random bits of different nodes are fully independent of each other. We show that limited independence suffices. Concretely, it is enough if each node has $\poly(\log n)$ bits of randomness and over the whole network these bits are $\poly(\log n)$-wise independent. Notice that this also allows us to say that we can build such a network decomposition with only $\poly(\log n)$ bits of shared randomness. The reason is that by standard constructions of $k$-wise independent random bits\cite{alon2004probabilistic}, we need only $O(k\log n)$ fully independent random bits to be able to produce $\poly(n)$ random bits that are $k$-wise independent. Hence, if there are $\poly(\log n)$ bits of shared randomness, we can construct $\poly(n)$ many bits of randomness out of them which are $\poly(\log n)$-wise independent, using a deterministic procedure. These bits can be shared among the nodes, e.g., just by identifiers.
Hence, we then have a setting where each node $\poly(n)$ bits of randomness and the bits over the whole network are $\poly(\log n)$-wise independent. By what we will show below, this allows us to build a network decomposition with $\poly(\log n)$ parameters in $\poly(\log n)$ rounds of the $\mathsf{LOCAL}$ model. 

\begin{theorem}
\label{thm:NetDecomp-with-K-wise} There is a distributed algorithm that builds a network decomposition with $\poly(\log n)$ parameters in $\poly(\log n)$ rounds of the $\mathsf{LOCAL}$ model, assuming each node has $\poly(\log n)$ bits of randomness and over the whole network these bits are $\poly(\log n)$-wise independent.
\end{theorem}
\begin{proof}
  To construct a network decomposition with $\poly(\log n)$ parameters, we instead give a $\poly(\log n)$-round algorithm for a different problem called \emph{conflict-free hypergraph multi-coloring}. In \cite{ghaffari2017complexity}, it is shown that network decomposition can be formulated as an instance of \emph{conflict-free hypergraph multi-coloring}. Hence, our solution for the latter problem immediately implies a construction for network decomposition.

In the \emph{conflict-free hypergraph multi-coloring}, we are given a hypergraph with $\poly(n)$ hyperedges, on the $n$ nodes of our graph. Moreover, hyperedges are grouped in $\log n$ classes, where all hyperedges of the $i^{th}$ class contain a number of vertices in $[2^{i-1}, 2^{i})$. The objective is to multi-color the vertices with $\poly(\log n)$ colors --- multi-coloring means a node is allowed to have many colors --- such that for each hyperedge, there is one color such that exactly one node of this hyperedge has that color. Ghaffari et al.\cite{ghaffari2017complexity} also gave a $\poly(\log n)$-round \emph{deterministic} algorithm for conflict-free hypergraph multi-coloring, whenever all hyperedges have size at most $\poly(\log n)$. We explain that a source of $\Theta(\log^2 n)$-wise independent random bits allows us to reduce the general case to this special case where all hyperedges have size at most $\poly(\log n)$, which can be solved then deterministically by the algorithm of \cite{ghaffari2017complexity}.

We use different colorings for the $\log n$-different hyperedge size classes. Our focus is now to get a conflict-free multi-coloring with $\poly(\log n)$ colors, for each hyperedge size class. Let us focus on one class: For hyperedges with size in $[2^{i-1}, 2^{i})$, if $2^{i}\leq \poly(\log n)$, we do not need to do anything as the deterministic algorithm of \cite{ghaffari2017complexity} is directly applicable. Otherwise, mark each node randomly with probability $p=\frac{\Theta(\log n)}{2^{i}}$. Notice that each node can do this using $\log n$ bits of randomness. Moreover, if these bits are $(k\log n)$-wise independent, then any set of at most $k$ nodes are marked independently of each other. That is, whether nodes are marked or not are $k$-wise independent, which means for any set $S$ of nodes with $|S|\leq k$, the probability that all nodes of $S$ are marked is $p^{k}$. We set $k=\Theta(\log n)$. Hence, by an application of the extension of Chernoff bound to $k$-wise independent random variables\cite{schmidt1995chernoff}, we can infer that in each hyperedge of size $[2^{i-1}, 2^{i})$, we have $\Theta(\log n)$ marked nodes, with probability $1-1/\poly(\log n)$. Now, we apply the deterministic algorithm of  \cite{ghaffari2017complexity} on these hyperedges, which consist of only marked nodes, hence getting a conflict-free multi-coloring of it with $\poly(\log n)$ colors, in $\poly(\log n)$ rounds. Finally, we recall that we use $\log n$ different collections of colors, for different hyperedge size classes, where each collection has $\poly(\log n)$ colors. Hence, we get a conflict-free multi-coloring of the whole hypergraph.
\end{proof}

An idea similar to above can be used to show that $\poly(\log n)$ bits of shared randomness suffice for network decomposition, thanks to standard constructions of $k$-wise independent random bits. The following theorem given an even stronger result, by showing that such a construction can be performed even in the $\mathsf{CONGEST}$ model (using a very different method, and by making use of a randomized construction of Elkin and Neiman\cite{elkin16_decomp}.)

\begin{theorem}
\label{thm:NetDecomp-with-Shared-Rand} There is a distributed algorithm that builds a $(O(\log n), O(\log^2 n))$-network decomposition with congestion $1$, in $\poly(\log n)$ rounds in the $\mathsf{CONGEST}$ model, using only $\poly(\log n)$ bits of shared randomness (and no private randomness).
\end{theorem}
\begin{proof}
  The construction works in $O(\log n)$ phases, where per phase each node gets clustered with at least a constant probability, and remains for the next phases otherwise. We describe the process for one phase, which will define non-adjacent clusters, each with radius $O(\log^2 n)$, such that each node is clustered with at least a constant probability. Repeating this process for $O(\log n)$ phases gives us the desired network decomposition, with high probability. The point for us is to describe how we perform one phase in $\poly(\log n)$ rounds of the $\mathsf{CONGEST}$ model and using only $\poly(\log n)$ bits of shared randomness. We first describe a construction assuming fully independent random bits and then argue why it can be performed using only $\poly(\log n)$ bits.

\paragraph{Construction for One Phase:} The construction for one phase consists of $p=\Theta(\log n)$ epoch, indexed by $i = 1, 2, ..., p$.
For epoch $i$, we define a base radius $R_i := (p-i) \cdot c\log n$ for a 
sufficiently large constant $c\geq 10$. 
%We choose the constant $c$ such that we can safely assume that each of the Geom(1/2) variables we draw is smaller than c*log n.
At the beginning of an epoch, each node that is still available 
decides to be a center with probability $O(2^i \cdot \log n)/n$. Each node $u$ that decides to be a center then chooses a random variable $X_u$ according to a geometric distribution with parameter $1/2$. That is, $Pr[X_u=Z]=2^{-Z}$ for any integer $Z\geq 1$. Notice that with high probability, for each center $u$, we have $X_u\leq c\log n$. 
The intuitive way of interpreting these random variables is that the cluster of center $u$ will grow up to a distance of at most $R_i+X_u$, from its center $u$. We say the cluster of $u$ can reach node $v$ if $(R_i+ X_u)\geq d_G(u,v)$. The actual clusters are defined as follows: Each node $v$ selects a center $u$ for which $(R_i+ X_u)- d_G(u,v) \geq 0$ and who maximizes $(R_i+ X_u)- d_G(u,v)$. For $v$, let $m_1$ and $m_2$ be the the two largest values $R_i - (X_u + d_G(u,v))$ among cluster center whose cluster reaches $v$. If there is no second cluster, define $m_2=0$. If $v$ received from at least one cluster center, then $v$ will be removed in this phase, and either \emph{clustered} or \emph{set aside}, according to the following criterion: If $m_1-m_2> 1$, then $v$ joins the cluster of its center $u$; otherwise, it is set aside and it remains unclustered for the whole duration of this phase (it will be brought back in the next phases, to get clustered then). If no cluster reaches $v$ in this epoch, then $v$ continues to the next epoch.

\paragraph{Correctness:}The fact that the carved clusters are non-adjacent and each has strong diameter at most $O(\log^2 n)$ are similar to the case discussed before, for the construction in \Cref{lem:CONGEST-MPX-BallCarving}; the argument is the same as \cite[Lemma 4]{elkin16_decomp}. We can also show that each node $v$ has probability $1/2$ to be clustered, in each phase. Consider the first epoch in which node $v$ is reached by some cluster. Similar to \Cref{lem:CONGEST-MPX-BallCarving}, and as shown in \cite[Claim 6]{elkin16_decomp}, we see that conditioned on $v$ being reached by some cluster, the probability that it has $m_1-m_2> 1$ is at least a constant. If $v$ is reached in this epoch but it has $m_1-m_2\in\{1, 0\}$, then it is set aside and it remains for the next phases. Given the fact that the probability of being chosen as a centered grows as $\frac{2^{i}\log n}{n}$ with epoch number $i$, in some epoch, some center's cluster reaches $v$; because at the very latest, in the last epoch, $v$ itself (if remaining) becomes a center with probability $1$ and its ball reaches itself. Thus, in each phase, $v$ gets clustered with at least a constant probability. Therefore, in $O(\log n)$ phases, node $v$ gets clustered with high probability.

\paragraph{Randomness:} We now discuss how to use a limited amount of randomness for the above algorithm. We use a source of $\Theta(\log^2 n)$-wise independent random bits for nodes deciding whether they are sampled or not, for each epoch. Randomness in different epochs and different phases are independent. Moreover, we use another (independent) source of $\Theta(\log^2 n)$-wise independent random bits for each sampled center $u$ determining their random radii $X_u$. We next argue why this suffices.

Notice that the radius $R_i$ decreases by $c\log n$ in each epoch, and the random values $X_u$ for the centers are always upper bounded by $c\log n$, with high probability. We can use this to conclude that in each epoch, w.h.p., each node $v$ can only be reached by at most $O(\log n)$ 
different centers. The reason is as follows: If there are more than $C\log n$ centers that can reach $v$ in the current epoch $i$, for a sufficiently large $C$, it means that the number of nodes in distance $R_{i-1}$ of $v$ is at least $\frac{Cn}{2 \cdot 2^{i}}$, with high probability (using Chernoff for sum of variables that are $(\log n)-$wise independent). That means, in epoch $i-1$ where the sampling probability was $\frac{\log n \cdot 2^{i-1}}{n}$, with high probability, at least one of these nodes should have been sampled to be a center (again, using Chernoff for sum of variables that are $(\log n)-$wise independent). That means, at least one center would have reached $v$ in the previous epoch, which means $v$ would have been removed for this epoch. Thus, we conclude that in each epoch, at most $O(\log n)$ sampled centers can reach $v$.

Now, we come to analyzing the random radii, and the probability of a node $v$ being clustered in the the first epoch in which a cluster reaches it. Notice that thanks to the property discussed above, there are only $O(\log n)$ sampled centers that can reach $v$. Hence, the event of whether $v$ is clustered or not depends only on the random radii of these $O(\log n)$ cluster centers. But that is an event that is determined by $O(\log^2 n)$ bits of randomness, $O(\log n)$-bits per cluster center. Given $\Theta(\log^2 n)$-wise independence of the bits used for defining radii, this space is the same as when the randomness is fully independent. Hence, the analysis explained above applies and shows that $v$ is clustered in this epoch with probability at least a constant. Since different phases have independent randomness, we can conclude that $v$ gets clustered in some phase, with high probability.
\end{proof}

\begin{theorem}\label{thm:onebit-perh-hop2}
Suppose that nodes $S\subseteq V$ of the network hold some independent bits of randomness, each holding just a single bit, and for each node $v\in V$, there is at least one node $s\in S$ within distance $h$ hops of $v$, where $h=\poly(\log n)$. There is a distributed algorithm that works in $h\poly(\log n)$ rounds of the $\mathsf{CONGEST}$ model and 
constructs a strong-diameter network decomposition with $O(\log n)$ colors and $O(\log^2 n)$ radius.
\end{theorem}
\begin{proof}[\textbf{Proof Sketch of \Cref{thm:onebit-perh-hop2}}]
We perform the bit gathering of Lemma \Cref{lem:BitGathering} so that each cluster center has $O(\log^4 n)$ bits. Then, we share this randomness to all the nodes of the cluster. We then apply the network decomposition explained in the previous section, using these $O(\log^4 n)$ bits. Notice that in each cluster, the bits of randomness are independent and they can be viewed as shared randomness. In different clusters, we have bits that are fully independent of each other. This allows us to run the construction of \Cref{subsec:sharedRand}, in a direct way, and obtain a strong-diameter network decomposition with $O(\log n)$ colors and $O(\log^2 n)$ radius.
\end{proof}

%\newpage
%%% Local Variables:
%%% mode: latex
%%% TeX-master: "main"
%%% End:

\section{Time vs.\ Error Probability Trade-Offs}
\label{sec:errorprob}

In this second part, we investigate the success
probability randomized local distributed graph algorithms. In
particular, we are interested in understanding the trade-off between
the time complexity of a randomized algorithm and its success
probability. Given a time budget, what is the best achievable
error probability and at what point is it even possible to completely
derandomize and obtain a deterministic algorithm? We first state a 
basic (and known) such derandomization lemma.

\begin{lemma}[Implicit in \cite{chang16}]\label{lemma:basicderand}
  Assume that we are given a non-uniform randomized distributed
  algorithm $\calA$ that is given $n$ as input and solves a given
  distributed graph problem $\calP$ on node graphs of size at most $n$
  with probability at least $1-2^{-n^2}$ in time $T(n)$. Then, there
  also exists a deterministic algorithm $\calA'$ that solves $\calP$
  on graphs of size at most $n$ in time $T(n)$. If $\calA$ works in
  the \CONGEST model, $\calA'$ also works in the \CONGEST model.
\end{lemma}
\begin{proof}
    We only sketch a proof of this lemma here. A more detailed proof for
  example appears as part of Theorem 3 in \cite{chang16}. 

  In a randomized algorithm, we can w.l.o.g.\ assume that each node
  $v\in V$ first chooses a sequence of random bits $X_v$ and that
  afterwards, the algorithm is run deterministically. We can think of
  this initial randomness as a random function $\phi(i)$ that assigns
  a random bit string to every possible node ID $i$. Assume that node
  IDs are from the range $\set{1,\dots,n^c}$ and let $\mathcal{G}_n$
  be the family of graphs with at most $n$ nodes, where each node has
  a unique label from $\set{1,\dots,n^c}$. The number of such graph is
  $|\mathcal{G}_n| \leq n\cdot 2^{{n\choose 2}}\cdot n^{cn} < 2^{n^2}$
  for sufficiently large $n$. If each of the possible assignments of
  random bits fails on at least one of the graphs in $\mathcal{G}_n$,
  the success probability of the algorithm cannot be better than
  $1-1/|\mathcal{G}_n|<1-2^{-n^2}$. There therefore needs to be at
  least one possible assignment of the random bits for each node ID
  such that the algorithm works for every graph in $\mathcal{G}_n$. We
  can run $\calA$ with this assignment of random bits and obtain a
  deterministic algorithm $\calA'$ with the same running time. The
  maximum messages size of $\calA'$ is at most as large as the maximum
  message size of $\calA$ and $\calA'$ therefore works in the \CONGEST
  model if $\calA$ works in the \CONGEST model.
\end{proof}

We are usually interested in success probabilities of the form
$1-1/n^c$ for some constant $c$. By increasing the running time by a
factor $\tau>1$, one can often improve the error probability by a
factor that is exponentially small in $\tau$. This leads to error
probabilities that are at best exponentially small in the running time. As a first contribution of this section,
the following theorem shows that by using the graph shattering
technique (cf.\ the respective discussion in \Cref{sec:intro}), we can
boost the success probability when computing a network
decomposition significantly beyond this.  By the reductions in
\cite{ghaffari2017complexity,ghaffari2018derandomizing}, we then also
immediately get the same improvement in the success probability for
all $\poly(\log n)$-locally checkable graph problems that have
$\poly(\log n)$-time randomized distributed algorithms with only
polynomially small error probability.

\begin{theorem}\label{thm:ep_alg}
  Let $T(n) \geq \log^c n$ for a sufficiently large constant
  $c>0$. There is a randomized \CONGEST model algorithm that computes
  a $\big(T(n), T(n)\big)$-decomposition with congestion $1$ and a
  randomized \LOCAL algorithm to compute a strong diameter
  $\big(O(\log n), O(\log n)\big)$-decomposition in time
  $T(n)$ with success probability at least
  $1 -n^{-2^{\eps\cdot\log^2 T(n)}}$ for some constant $\eps>0$. 
\end{theorem}
\begin{proof}
    For convenience, we only prove that the time complexities of the
  resulting \CONGEST and \LOCAL algorithms is $T(n)\cdot \poly\log n$.
  Note that by applying this weaker statement for an appropriately
  chosen smaller value of $T(n)$ and by choosing the constants $c$ and
  $\eps$ appropriately, the claim of the theorem then follows.

  The algorithm consists of 2 main steps. First, we apply a standard
  randomized polylog-time network decomposition algorithm
  $\mathcal{A}$ that succeeds w.h.p. The success probability of
  $\mathcal{A}$ is not as high as we need it to be, however, we show
  that after running $\mathcal{A}$, the number of sufficiently
  separated remaining nodes is small with the ``right'' success
  probability. This allows to compute a decomposition on the remaining
  nodes by using adeterministic network decomposition algorithm.

  For the first step, we apply the randomized decomposition algorithm
  of Elkin and Neiman \cite{elkin16_decomp}, which computes a strong
  diameter $\big(O(\log n), O(\log n)\big)$-decomposition in time
  $O(\log^2 n)$ in the \CONGEST model (and thus in particular a
  decomposition with congestion $1$). The parameters in the algorithm
  of \cite{elkin16_decomp} such that it succeeds with probability at
  least $1-1/n^2$. In the following, assume that the algorithm is run
  in this way. In particular, this also implies that for every node
  $v\in V$, the probability at $v$ is in a cluster after running the
  algorithm is at leat $1-1/n^2$.

  For a subset $S\subseteq V$, we say that $S$ is $d$-separated if any
  two nodes in $S$ are at distance at least $d$ in $G$. Let
  $t(n)=O(\log^2 n$) be the running time of the algorithm of
  \cite{elkin16_decomp}. Further, let $\bar{V}\subseteq V$ be the set
  of nodes that are not inside a cluster after running the
  algorithm. Note that the output of a node $u\in V$ can only depend
  on the random bits of nodes within distance at most $t(n)$ of
  $u$. Hence, the outputs of a set $S$ of nodes at pairwise distance
  at least $2t(n)+1$ are independent. In particular, the events that
  the nodes in $S$ are in $\bar{V}$ are independent. Let
  $S\subseteq V$ therefore be such a $(2t(n)+1)$-separated set of
  nodes. The probability that all nodes in $S$ are in $\bar{V}$ is at
  most $1/n^{2|S|}$. For any $K\geq 1$, the probability that $\bar{V}$
  contains a $(2t(n)+1)$-separated set $S$ of size at least $K$ is
  thus at most ${n\choose K}\cdot 1/n^{2K} \leq 1/n^K$. We choose the
  value $K$ such that this probability is upper bounded by the failure
  probability required by the lemma statement, i.e., we choose $K$
  such that $n^K \geq n^{2^{\eps\cdot \log^2 T(n)}}$ and we can thus
  choose $K = 2^{\eps\cdot \log^2 T(n)}$ for an appropriate constant
  $\eps>0$.

  Our goal is to reduce the remaining problem to deterministically
  computing a network decomposition on a graph of size at most $K$.
  As a first step, we compute a $(2t(n)+1)$-separated subset
  $S\subseteq\bar{V}$ of remaining nodes.  This can be done by
  computing a $\big(2t(n)+1, O(t(n)\log n)\big)$-ruling set $S$
  w.r.t.\ $\bar{V}$. Note that because $S$ is a $(2t(n)+1)$-separated
  set of nodes, with probability at least $1-1/n^K$, its size is at
  most $K$. For the rest of the proof, we therefore assume that $S$ is
  of size at most $K$. Because $S$ is a
  $\big(2t(n)+1, O(t(n)\log n)\big)$-ruling set w.r.t.\ $\bar{V}$,
  each remaining node $u\in \bar{V}$ has at least one node in $S$
  within distance $O(t(n)\log n)$ in $G$. By starting parallel BFS
  explorations from each node $v\in S$, we can therefore build a
  cluster $C_v$ of radius at most $O(t(n)\log n)$ around each node
  $v\in S$ such that each node $u\in \bar{V}$ is contained in the
  cluster of its closest node in $S$ (ties broken arbitrarily). For
  each cluster $C_v$, we obtain a spanning tree of depth
  $O(t(n)\log n)$ and the trees of different clusters are
  vertex-disjoint. We note that the trees might contain nodes of
  $V\setminus\bar{V}$. We define the cluster graph $G_C$ as the graph
  defined on the set of clusters $C_v$ for $v\in S$, where two
  clusters $C_u$ and $C_v$ are neighbors whenever there are nodes
  $x\in C_u\cap\bar{V}$ and $y\in C_v\cap\bar{V}$ that are neighbors
  in $G$. In \cite{ghaffari19_MIS}, it is possible to compute a
  strong-diameter
  $\big(2^{O(\sqrt{\log K})}, 2^{O(\sqrt{\log K})}\big)$-decomposition
  of such a cluster graph in time $2^{O(\sqrt{\log K})}$ times the
  maximum cluster radius. Note that since the clusters a
  vertex-disjoint, a strong-diameter decomposition of $G_C$ leads to a
  decomposition with congestion $1$ for the remaining nodes $\bar{V}$
  on $G$. By our choice of $K=2^{\eps\cdot \log^2 T(n)}$, we have
  $2^{c\sqrt{\log K}} = T^{\sqrt{\eps}\cdot c}$ and the \CONGEST model
  part of the claim of the lemma thus follows by choosing the constant
  $\eps>0$ sufficiently small. The claim about the \LOCAL model then
  follows in the \LOCAL model any $(d(n),c(n))$-decomposition can be
  turned into an $(O(\log n), O(\log n))$-decomposition in time
  $O(d(n)\cdot c(n)\cdot \log^2n)$
  \cite{awerbuch96,ghaffari2017complexity}.
\end{proof}

The following theorem shows that the error probability bound of
\Cref{thm:ep_alg} is essentially tight, in the following sense. A
stronger bound would be a major breakthrough and directly
imply a significant improvement over the 25-year old and
currently best deterministic network decomposition algorithm~\cite{panconesi95}, which has a time
complexity of $2^{O(\sqrt{\log n})}$. The following theorem also shows
that the approach of the above algorithm is best possible. A better
deterministic network decomposition algorithm resulting from a better
error probability bound would directly also allow to improve the above
algorithm to achieve this same error probability bound. In order to
work for general locally checkable problems (such that in particular
network decomposition is included), the following theorem is stated in
an rather technical way. We discuss what it means for specific
problems below.

\begin{theorem}\label{thm:ep_derandomization}
  Let $\calP(n)$ be a distributed graph problem, where the solution
  might depend on $n$. Assume that there is non-uniform randomized
  distributed algorithm that given $n$ as input solves $\calP(n)$ in
  time $T(n)$ with probability at least
  $1-2^{2^{\eps\log^\beta T(n)}}$ for some constant $\eps>0$ and some
  $\beta > 2$.  Then, there is a deterministic distributed algorithm
  that solves $\calP(N)$ in time $2^{O(\log^{1/\beta} n)}$, where $N$
  is chosen such that $T(N)=2^{c\cdot\log^{1/\beta} n}$ for some
  constant $c>0$. 

  If the problem $\calP(n)$ is strictly $d(n)$-locally checkable for
  some function $d(n)$, the deterministic algorithm solves $\calP(n)$
  rather than $\calP(N)$. Further, if the randomized algorithm works
  in the \CONGEST model, the deterministic algorithm also works in the
  \CONGEST model.
\end{theorem}
\begin{proof}
    We apply a technique that was first used by Chang, Kopelowitz,
  and Pettie in \cite{chang16}. The basic idea is to ``lie'' to the
  algorithm about the number of nodes and pretend that the graph has
  size $N\gg n$ rather than $n$ in order to boost the success
  probability of the algorithm to value that is sufficiently clost to
  $1$ in order to apply \Cref{lemma:basicderand}.

  Assume that we are given a graph $G$ with $n$ nodes on which we have
  to solve $\calP(n)$. The nodes of $G$ cannot distinguish the
  graph $G$ from a graph $G'$ with $N$ nodes, which contains $G$ as
  one of its connected components. A randomized algorithm that is
  given $n$ as input and solves $\calP(n)$ on $G$ in time $T(n)$
  with probability at least $1-\delta(n)$ can therefore also be used
  to solve $\calP(N)$ on $G'$ (and thus also on $G$) in time
  $T(N)$ with probability at least $1-\delta(N)$. If we choose $N$
  such that $\delta(N)\leq 2^{-n^2}$ by \Cref{lemma:basicderand}, the
  resulting algorithm can be derandomized when applied on $G$.

  We thus have to choose the size $N$ of the ``virtual'' graph $G'$
  such that $2^{\eps\log^{\beta}T(N)}\geq n^2$. This implies that
  $\log T(N) \geq \big(\frac{2}{\eps}\big)^{1/\beta}\log^{1/\beta}n$,
  which is satisfied for $T(N)=2^{c\cdot \log^{1/\beta}n}$ for some
  constant $c>0$. We therefore get a deterministic algorithm for
  $\calP(N)$ with time complexity $2^{O(\log^{1/\beta}n)}$,
  where $N$ is chosen such that $T(N)=2^{c\cdot \log^{1/\beta}N}$. If
  the randomized algorithm works in the \CONGEST model, the
  deterministic algorithms also works in the \CONGEST model.

  If the problem $\calP(n)$ is strictly $d(n)$-locally checkable for
  some function $d(n)$, the checking algorithm cannot distinguish
  between the graphs $G$ and the component isomorphic to $G$ in
  $G'$. A valid solution to $\calP(N)$ on a graph of size $N$ is
  therefore also a valid solution to $\calP(n')$ for each connected
  component of $G$ of size $n'$.
\end{proof}

As the above theorem is stated for general graph problems $\calP(n)$
and general time complexities $T(n)$ (which both can depend on $n$ in
various ways), it is interesting to discuss what the theorem implies
for specific choices of $T(n)$ and for concrete graph problems. We
first consider the case of randomized algorithms that have running
times that are between $\poly(\log n)$ and $2^{O(\sqrt{\log
    n})}$.
That is, we consider running times that are between the best
randomized (in the classic sense) and deterministic time complexities
for network decomposition and many other important problems. The
following corollary follows immediately from
\Cref{thm:ep_derandomization}.

\begin{corollary}\label{cor:specificTn}
  Let $\alpha>\beta>2$ be two constants, let $\calP(n)$ be a
  distributed graph problem, and assume that there is non-uniform
  randomized distributed algorithm that given $n$ as input solves
  $\calP(n)$ in time $2^{O(\log^{1/\alpha}n)}$ with probability at
  least $1-2^{2^{\eps\log^{\beta/\alpha} n}}$ for some constant
  $\eps>0$.  Then, there is a deterministic distributed algorithm that
  solves $\calP(N)$ in time $2^{O(\log^{1/\beta} n)}$, where $N$ is
  chosen such that $\log N = \Theta(\log^{\alpha/\beta}n)$ for some constant
  $c>0$. 

  If the problem $\calP(n)$ is strictly $d(n)$-locally checkable for
  some function $d(n)$, the deterministic algorithm solves $\calP(n)$
  rather than $\calP(N)$. Further, if the randomized algorithm works
  in the \CONGEST model, the deterministic algorithm also works in the
  \CONGEST model.
\end{corollary}

Note that if the parameters $\calP(n)$ only depend polylogarithmically
on $n$, the parameters of $\calP(N)$ also
depend polylogarithmically on $n$. For example if $\calP(n)$ is the
problem of computing an $\big(O(\log n),O(\log n)\big)$-network
decomposition, then $\calP(N)$ refers to the problem of computing an
$\big(O(\log^{\alpha/\beta} n),O(\log^{\alpha/\beta} n)\big)$-network decomposition.
Further, if $\calP$ is an LCL problem (e.g., MIS or $(\Delta+1)$-coloring), $\calP$ is strictly $O(1)$-locally checkable. Hence, for this important and
widely studied class of problems, the stronger versions of
\Cref{thm:ep_derandomization} and \Cref{cor:specificTn} hold, i.e., in
both cases we get a deterministic $2^{O(\log^{1/\beta} n)}$-time
algorithm that solves $\calP$ in its original form. The following
corollary shows that the same (and actually something slightly
stronger) also holds for the network decomposition problem in the
\LOCAL model.

\begin{corollary}\label{cor:decomposition}
  Assume that there is non-uniform randomized distributed algorithm
  that given $n$ as input computes a $(T(n),T(n))$-decomposition in
  time $T(n)$ with probability at least
  $1-2^{2^{\eps\log^\beta T(n)}}$ for some constant $\eps>0$ and some
  $\beta > 2$. Then, there exists a deterministic
  $2^{O(\log^{1/\beta}n)}$-round algorithm to compute a
  $\big(O(\log n), O(\log n)\big)$-decomposition in the \LOCAL model.
\end{corollary}
\begin{proof}
    \Cref{thm:ep_derandomization} directly implies that there exists a
  deterministic $T(N)$-round algorithm to compute a
  $(T(N),T(N))$-decompostion, where $T(N) =
  2^{O(\log^{1/\beta}n)}$. In \cite{awerbuch96}, it is shown that
  given such a decomposition, one can compute a strong diameter
  $\big(O(\log n), O(\log n)\big)$-decomposition in time
  $2^{O(\log^{1/\beta}n)}$ deterministically in the \LOCAL model.
\end{proof}

The next theorem bounds the error probability that is needed to
derandomize a polylogarithmic-time randomized algorithm to a
polylogarithmic-time deterministic algorithm. As this result is most
interesting for strictly locally checkable graph problems, we only
state it for this case.

\begin{theorem}\label{thm:polylogderand}
  Let $\calP$ be a strictly $d(n)$-locally checkable graph problem for
  some function $d(n)$ or let $\calP$ be the problem of computing a
  network decomposition with $\poly(\log n)$ parameters. Further
  assume that there is randomized non-uniform distributed algorithm
  that solves $\calP$ in $\poly(\log n)$ time with probability at
  least $1-2^{-2^{\log^\eps n}}$ for some constants $\eps>0$. Then,
  there is a deterministic distributed $\poly\log n$-time algorithm
  for $\calP$.
\end{theorem}
\begin{proof}
    We use the same basic technique as in the proof of
  \Cref{thm:ep_derandomization}. We again lie to the randomized
  algorithm about the number of nodes and pretend that the number of
  nodes is $N\gg n$. We choose $N$ such that $2^{\log^\eps N}\geq n^2$
  and thus $\log N \geq (2\log n)^{1/\eps}$. A running time that is
  polylogarithmic in $N$ is therefore also polylogarithmic in $n$,
  which proves the theorem. The result for network decompositions
  follows in the same way as in \Cref{cor:decomposition}.
\end{proof}

\paragraph{Remark:}
One could of course get similar results for other time complexity
bounds. In the same way as in \Cref{thm:polylogderand}, it can for
example be shown that if we have a randomized quasi-polylogarithmic
time algorithm with success probability
$1-2^{-2^{2^{(\log\log n)^{\eps}}}}$ for some constant $\eps>0$, we
can derandomize it to a quasi-polylogarithmic deterministic algorithm
for the same problem. By quasi-polylogarithmic running time, we mean a
running time of the form $2^{(\log\log n)^c}$ for some constant
$c>0$.

%%% Local Variables:
%%% mode: latex
%%% TeX-master: "main"
%%% End:

\bibliography{references}
\bibliographystyle{alpha}
\appendix
\end{document}